\documentclass[11pt]{article}
\usepackage{amsthm}

\usepackage[latin1]{inputenc}
\usepackage{enumerate}
\usepackage{amsmath}
\usepackage{dsfont}
\usepackage{vaucanson-g}
\usepackage{wasysym}

\newtheorem{thm}{Theorem}
\newtheorem{prop}{Proposition}
\newtheorem{lem}{Lemma}
\newtheorem{cor}{Corollary}

\newcommand{\fado}{{\bf F{\cal A}do}}

\newcommand{\lixo}[1]{}
\newcommand{\ICDFA}{\textsc{ICDFA}}
\newcommand{\ICDFAE}{$\mbox{\textsc{ICDFA}}_\emptyset$}
\newcommand{\DFAE}{$\mbox{\textsc{DFA}}_\emptyset$}
\newcommand{\DFA}{\textsc{DFA}}
\newcommand{\floor}[1]{\lfloor #1 \rfloor}
\newcommand{\djk}[2]{(\floor{#1/#2},\sigma_{#1\bmod{#2}})}

\begin{document}

\title{On the Representation of Finite Automata \thanks{Work partially
    funded by Funda\c{c}\~ao para a Ci\^encia e Tecnologia (FCT) and
    Program POSI.}\footnote{This paper was presented at the 7th
    Workshop on Descriptional  Complexity of Formal Systems, DCFS'05}}

\author{Rogério Reis\\
  {\tt rvr@ncc.up.pt}\\
  \and Nelma Moreira
  \thanks{Corresponding author}\\
  {\tt nam@ncc.up.pt}\\
  \and
  Marco Almeida\\
  {\tt mfa@alunos.dcc.fc.up.pt}\\
  \and DCC-FC \ \& LIACC,  Universidade do Porto \\
  R. do Campo Alegre 823, 4150 Porto, Portugal} \date{}

\maketitle

\begin{abstract}
  We give an unique string representation, up to isomorphism, for
  initially connected deterministic finite automata (\ICDFA{}'s) with
  $n$ states over an alphabet of $k$ symbols.  We show how to generate
  all these strings for each $n$ and $k$, and how its enumeration
  provides an alternative way to obtain the exact number of
  \ICDFA{}'s.
\end{abstract}
\section{Motivation}
In symbolic manipulation environments for finite automata, it is
important to have an adequate representation of automata and,
dependent upon their use, several representations may be available.
For example, for testing if two finite automata are isomorphic objects
or for (random) generation of automata, the representation must be
compact and somehow canonical. In the \fado{}
project~\cite{moreira05:_inter_manip_regul_objec_fado,fado} a
canonical form is used to test if two minimal \DFA{}'s are
\emph{isomorphic} (i.e are the same up to renaming of states). In this
paper we prove the correctness of that representation and show how it
can be used for the exact enumeration and generation of initially
connected deterministic finite automata (\ICDFA{}).
The problem of enumeration of finite automata was considered by
several authors since early 1960s, in particular see
Robinson~\cite{robinson85:_count}, Harary and Palmer
\cite{harary67:_enumer}\lixo{harary73:_graph_enumer} and
Liskovets~\cite{liskovets69} amongst many others. A survey may be
found in Domaratzki et al.~\cite{domaratzki02}. More recently, several
authors examined related problems. Domaratzki et
al.~\cite{domaratzki02} studied the \lixo{(exact and asymptotic)
}enumeration of distinct languages accepted by finite automata with
$n$ states; Nicaud~\cite{nicaud99:_averag}, Champarnaud and
Parantho\"en~\cite{champarnaud:_random_gener_dfas,paranthoen04} and
Bassino and Nicaud~\cite{bassino:_enumer} analysed several aspects of
the average behaviour of regular languages;
Liskovets~\cite{liskovets03:_exact} and
Domaratzki~\cite{domaratzki04:_combin_inter_gener_genoc_number} gave
(exact and asymptotic) enumerations of acyclic \DFA{}'s and of finite
languages.
The paper is organised as follows. In the next section, we review some
basic notions and introduce some notation.
Section~\ref{sec:repr-towards-norm} describes a string representation
for deterministic finite automata that is unique up to isomorphism for
initially connected deterministic finite automata.
Section~\ref{sec:generating-automata} presents an efficient method to
generate those strings. Section~\ref{sec:enumeration} shows how their
enumeration provides an upper bound and the exact value for the number
of \ICDFA{}'s.
Section~~\ref{sec:implementation} concludes with some final remarks.
We address the reader attention to the longer version of this paper
for some implementation issues and experimental
results\footnote{http://www.dcc.fc.up.pt/Pubs/TR05/dcc-2005-04.ps.gz}.
\section{Preliminaries}
We first recall some basic notions from automata theory and formal
languages, that can be found in standard
books~\cite{hopcroft00:_introd_autom_theor_languag_comput}.  An
alphabet $\Sigma$ is a nonempty set of symbols. A string over $\Sigma$
is a finite sequence of symbols of $\Sigma$. The empty string is
denoted by $\epsilon$. The set $\Sigma^\star$ is the set of all
strings over $\Sigma$. A language $L$ is a subset of
$\Sigma^\star$. The density of a language L over $\Sigma$,
$\rho_{L}(n)$, is the number of strings of length $n$ that are in $L$,
i.e., $\rho_{L}(n)=|L\cap \Sigma^{n} |$.
If $L_1,L_2\subseteq \Sigma^\star$, $L_1 L_2=\{ xy\mid x\in L_1\text{
  and } y\in L_2\}$.
A regular expression (r.e.)  $\alpha$ over $\Sigma$ represents a
language $L(\alpha)\subseteq \Sigma^\star$ and is inductively defined
by: $\emptyset$, $\epsilon$ and $\sigma\in \Sigma$ are a r.e., where
$L(\emptyset)=\emptyset$, $L(\epsilon)=\{\epsilon\}$ and
$L(\sigma)=\{\sigma\}$; if $\alpha_1$ and $\alpha_2$ are r.e.,
$(\alpha_1 + \alpha_2)$, $(\alpha_1 \alpha_2)$ and $\alpha_1^\star$
are r.e., respectively with $L((\alpha_1 + \alpha_2))=L(\alpha_1)\cup
L(\alpha_2)$, $L((\alpha_1\alpha_2))=L(\alpha_1)L(\alpha_2)$ and
$L({\alpha_1}^\star)=L(\alpha_1)^\star$. In this paper, we will use
regular expressions to represent descriptions of finite automata.  A
\emph{deterministic finite automaton} (\DFA{}) ${\cal A}$ is a
quintuple $(Q,\Sigma,\delta,q_0,F)$ where $Q$ is a finite set of
states, $\Sigma$ is the alphabet, $\delta: Q \times \Sigma \rightarrow
Q$ is the transition function, $q_0$ the initial state and $F\subseteq
Q$ the set of final states. We assume that the transition function is
total, so we consider only \emph{complete} \DFA{}'s.  The \emph{size}
of a \DFA{} is the number of its states, $|Q|$. Normally, we are not
interested in the labels of the states and we can represent them by an
integer $0\leq i < |Q|$.  The transition function $\delta$ extends
naturally to $\Sigma^\star$: for all $q\in Q$, if $x=\epsilon$ then
$\delta(q,\epsilon)=q$; if $x=y\sigma$ then
$\delta(q,x)=\delta(\delta(q,y),\sigma)$.  A \DFA{} is \emph{initially
  connected}\footnote{Also called \emph{accessible}.}  (\ICDFA{}) if
for each state $q\in Q$ there exists a string $x\in \Sigma^\star$ such
that $\delta(q_0,x)=q$.  Two \DFA{}'s ${\cal
  A}=(Q,\Sigma,\delta,q_0,F)$ and ${\cal
  A}'=(Q',\Sigma,\delta',q_0',F')$ are called \emph{isomorphic} (by
states) if there exists a bijection $f:Q \rightarrow Q'$ such that
$f(q_0)=q_0'$ and for all $\sigma\in \Sigma$ and $q\in Q$,
$f(\delta(q,\sigma))=\delta'(f(q),\sigma)$.  Furthermore, for all
$q\in Q$, $q\in F$ if and only if $f(q)\in F'$.  The \emph{language}
accepted by a \DFA{} ${\cal A}$ is $L({\cal A})=\{x\in
\Sigma^\star\mid \delta(q_0,x)\in F\}$. Two \DFA{} are
\emph{equivalent} if they accept the same language. Obviously, two
isomorphic automata are equivalent, but two non-isomorphic automata
may be equivalent. A \DFA{}\ ${\cal A}$ is \textsl{minimal} if there
is no \DFA{}\ ${\cal A}'$ with fewer states equivalent to ${\cal
  A}$. Trivially a minimal \DFA{} is an \ICDFA{}. Minimal \DFA{}'s are
unique up to isomorphism. We are mainly concerned with the
representation of the transition function of \DFA{}'s
, so we disregard the set of final states and we consider only a
quadruple $(Q,\Sigma,\delta,q_0)$ called the \emph{structure} of an
automaton and referred as \DFAE. For each of our representations,
there will be $2^n$ \DFA{}'s. We denote by \ICDFAE{} the structure of
an \ICDFA{}.  We consider that any integer variable has always a
nonnegative value (if not otherwise stated). Let
$[n]_0=\{0,1,\ldots,n\}$ and $[n]=\{1,\ldots,n\}$.

\section{Representations towards a normal form}
\label{sec:repr-towards-norm}
The method used to represent a \DFA{} has a significative role in the
amount of computer work needed to manipulate that information, and can
give an important insight about this set of objects, both in its
characterisation and enumeration. Let us disregard
the set of \textit{final states} of a \DFA{}.  A \textit{naive}
representation of a \DFAE{} can be obtained by the enumeration of its states
and for each state a list of its transitions for each symbol. For the
\DFAE{} in Fig.\ref{fig:DFA1} we have:
  \begin{multline}
    [[A\; (\mathtt{a}:A, \mathtt{b}:B)],\,[B\;
    (\mathtt{a}:A,\mathtt{b}:E)], [C\;(\mathtt{a}:B, \mathtt{b}:E)],\\
    [D\;(\mathtt{a}:D, \mathtt{b}:C)], [E\;(\mathtt{a}:A,
    \mathtt{b}:E)]].
\label{eq:1}
  \end{multline}
\begin{figure}[h]
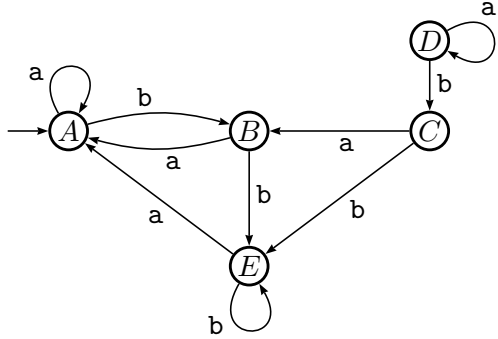

  \begin{center}
    \VCDraw{
      \begin{VCPicture}{(0,0)(12,6)}
        \State[A]{(2,4)}{A} \State[B]{(6,4)}{B} \State[D]{(10,6)}{C}
        \State[C]{(10,4)}{D} \State[E]{(6,1)}{E} \Initial{A}
        \LoopN{A}{\texttt{a}} \ArcL{A}{B}{\texttt{b}}
        \ArcL{B}{A}{\texttt{a}} \EdgeL{B}{E}{\texttt{b}}
        \LoopE{C}{\texttt{a}} \EdgeL{C}{D}{\texttt{b}}
        \EdgeL{D}{B}{\texttt{a}} \EdgeL{D}{E}{\texttt{b}}
        \LoopS{E}{\texttt{b}} \EdgeL{E}{A}{\texttt{a}}
      \end{VCPicture}}
  \end{center}\centering
  \caption{A \DFA{} with no final states marked}
  \label{fig:DFA1}
\end{figure}

Given a complete \DFAE{} $(Q,\Sigma,\delta,q_0)$ with $|Q|=n$ and
$|\Sigma|=k$ and considering a total order over $\Sigma$, the
representation can be simplified by omitting the alphabetic symbols.
For our example, we would have
\begin{equation}
  \label{eq:2}
  [[A\;(A,B)],[B\;(A,E)],[C\;(B,E)],[D\;(D,C)],[E\;(A,E)]].
\end{equation}

The labels chosen for the states have a standard order (in the
example, the alphabetic order). We can simplify the representation a
bit if we use that order to identify the states, and because we are
representing complete \DFAE{}'s we can drop the inner tuples as
well. We obtain
\begin{equation}
  \label{eq:3}
  [0,1,0,4,1,4,3,2,0,4].
\end{equation}

Because this representation depends on the order we label the states,
we have more than one representation for each \DFAE{}. Can we have a
canonical order for the set of the states? Let the first state be the
initial state $q_0$ of the automaton, the second state the first one
to be referred (excepting $q_0$) by a transition from $q_0$, the third
state the next referred in transitions from one of the first two
states, and so on... For the \DFAE{} in the example, this method
induces an unique order for the first three states ($A,B,E$), but then
we can arbitrate an order for the remaining states ($C,D$). Two
different representations are thus admissible:
\begin{equation}
  [0,1,0,2,0,2,3,4,1,2]\text{ and } [0,1,0,2,0,2,1,2,4,3].
\end{equation}
If we restrict this representation to \ICDFAE{}'s, then this
representation is unique and defines an order over the set of its
states. In the example, the \DFAE{} restricted to  the set of states $\{A,B,E\}$ 
is represented by $[0,1,0,2,0,2]$.
Let $\Sigma = \{\sigma_i\mid i<k\}$, with
$\sigma_0<\sigma_1<\cdots<\sigma_{k-1}$.  Given an \ICDFAE{}
$(Q,\Sigma,\delta,q_0)$ with $|Q|=n$, the representing string is of
the form $[(S_i)_{i<kn}]$ with $S_i\in[n-1]_0$ and
$S_i=\delta(\floor{i/k},\sigma_{i\bmod{k}})$.
\begin{lem}\label{lem:ICDFA1}
  Let $[(S_i)_{i<kn}]$ be a representation of a complete \ICDFAE{}
  $(Q,\Sigma,\delta,q_0)$ with $|Q|=n$ and $|\Sigma|=k$, then:
  \begin{gather}
      (\forall m>1)(\forall i)(S_i=m\;\Rightarrow\;((\exists
      j<i)\,S_j=m-1))\tag{\textbf{R1}}\label{eq:4}\\
      (\forall m \in [n-1])((\exists j<km)\,S_j=m)\tag{\textbf{R2}}
    \end{gather}
\end{lem}

\begin{proof}
The condition~\textbf{R1} establishes that a state label (greater than
$1$) can only occur after the occurrence of its predecessors. This is
a direct consequence of the way we defined the representing string.
Suppose \textbf{R2} does not verify, thus there exists a state $m$
that does not occur in the first $km$ symbols of the string (the $m$
first state descriptions). Because the automaton is initially
connected there must be a sequence of states $(m_i)_{i\leq l}$ and
symbols $(\sigma_i)_{i\leq l}$ such that $m_0=0,\; m_l=m$ and
$\delta(m_i,\sigma_i) = m_{i+1}$ for $i<l$. We must have $0<m<m_{l-1}$
because $m$ appears in the $m_{l-1}$ description and we supposed no
occurrences of $m$ in the first $m$ state descriptions.  There must
exist $l'<l$ such that $m_{l'-1}<m<m_{l'}$, implying that $m_ {l'}\in
\{S_i\,\mid\,i<km\}$. This contradicts \textbf{R1} because we are
supposing that $m\not\in\{S_i\,\mid\,i<km\}$ and $m<m_{l'}$. Thus
\textbf{R2} is verified.
\end{proof}
Note that the conditions \textbf{R1} and \textbf{R2} are
independent. For $k=2$ and $n=3$, the string $[2,1,0,0,1,0]$ satisfies
\textbf{R2} but not \textbf{R1}, and the opposite occurs for the
string $[0,0,1,1,0,2]$.

\begin{lem}\label{lem:ICDFA2}
  Every string $[(S_i)_{i<kn}]$ with $S_i\in[n-1]_0$ satisfying
  \textbf{R1} and \textbf{R2} represents a complete \ICDFAE{}{} with
  $n$ states over an alphabet of $k$ symbols.
\end{lem}
\begin{proof}
  Let $[(S_i)_{i<kn}]$ be a string in the referred conditions, and
  consider the associated automaton ${\cal A}$ using the string
  symbols as labels for the corresponding states. By its construction,
  ${\cal A}$ is a
  \DFAE{}.
  We only need to prove that it is initially connected.  Let $m$ be a
  state of the automaton.
  A proof that $m$ is reachable from the initial state $0$ can be done
  by induction on $m$.
  If $m=0$ there is nothing to prove.  If $m=1$ then, by \textbf{R2},
  $1$ must occur in the description of state $0$, making state $1$
  reachable from state $0$.
  Let us suppose that every state $m'<m$ is reachable from state $0$
  and prove that state $m$ is reachable too.  By \textbf{R2}, $m$
  occurs at least once before position $km$, say in position $km'+i$
  with $m'<m$ and $i<k$. Then for some symbol $\sigma$,
  $\delta(m',\sigma) = m$. By induction hypothesis, state $m'$ is
  reachable from state $0$, thus state $m$ is reachable too and the
  automaton is initially connected.
  Now consider the string representation obtained for ${\cal A}$,
  $[(S'_i)_{i<kn}]$.  By Lemma~\ref{lem:ICDFA1} it satisfies
  \textbf{R1} and \textbf{R2}.  It is easy to see that this
  representation is the same as $[(S_i)_{i<kn}]$.  By \textbf{R1},
  $S_0=S'_0$.  Suppose that $(\forall i<j)(S_i=S'_i)$.  Now we prove
  that $S'_{j}= S_{j}$.  By \textbf{R1}, either $S_j\in\{S_i\mid
  i<j\}$ or $S_j=\max\{S_i\mid i<j\}+1$.  In the first case, there
  exists $l<j$ such that $S_j=S_l$ and, by induction hypothesis,
  $S_l=S'_l$, thus

Analogously, by
  \textbf{R1}, in the second case we have that 
$$S'_j=\max\{S_i\mid i<j\}+1=S_j$$.
\end{proof}
\begin{thm}
\label{thm:iso}
There is a one-to-one mapping between strings $[(S_i)_{i<kn}]$ with
$S_i\in[n-1]_0$ satisfying \textbf{R1} and \textbf{R2}, and the
non-isomorphic \ICDFAE{}'s with $n$ states, over an alphabet $\Sigma$
of size $k$.
\end{thm}
\begin{proof}
  Let $(Q,\Sigma,\delta,q_0)$ and $(Q',\Sigma,\delta',q'_0)$ be two
  \ICDFAE{}'s and $[(S_i)_{i<kn}]$ and $[(S'_i)_{i<kn}]$ their
  representing strings. By Lemma~\ref{lem:ICDFA1}, these strings
  satisfy \textbf{R1} and \textbf{R2}. Suppose that
  $f:Q\longrightarrow Q'$ is an isomorphism between the
  \ICDFAE{}'s. Then $0=q_0$ and $f(q_0)=q'_0=0$. Either
  $S_0=\delta(q_0,\sigma_0)=q_0=0$
  or $S_0=\delta(s_0,\sigma_0)=1$ (by \textbf{R1}).
  \begin{enumerate}[i)]
  \item If $S_0=0$ then $f(q_0)=f(\delta(q_0,\sigma_0)) =
    \delta'(q'_0,\sigma_0)=S'_0=0,$ because $\delta(q_0,\sigma_0)=q_0$
    implies $\delta'(f(q_0),\sigma_0)=f(q_0)$.
  \item  If $S_0=1$ then
    $f(1)=\delta'(q'_0,\sigma_0)=S'_0\neq 0,$ thus $S'_0=1$, again by
    \textbf{R1}. 
  \end{enumerate}
  Supposing that $(\forall i<j)(S_i=S'_i\wedge f(S_i)=S'_i)$ we need
  to prove that $S_j=S'_j\wedge f(S_j)=S'_j$.  Trivially we have
  $\{S_i\mid i<j\}=\{S'_i\mid i<j\}$. We know that $S_j =
  \delta\djk{j}{k}$, and by \textbf{R2} there exists $l<j$ such that
  $\floor{j/k}= S_l$ thus $f(\floor{j/k})=f(S_l)=S_l=S'_l=\floor{j/k}$
  by induction hypothesis. We have
$$S'_j \;=\;\delta'( \floor{j/k}, \sigma_{j\bmod{k}}) \;=\;\delta'(
f(\floor{j/k}), \sigma_{j\bmod{k}})\;=\; f(\delta\djk{j}{k})\;=\;
f(S_j).
  $$
  By \textbf{R1}, either $S_j\in\{S_i\mid i<j\}$ or $S_j=\max\{S_i\mid
  i<j\}+1$.
  \begin{enumerate}[i)]
  \item If $S_j\in\{S_i\mid i<j\}$ then there exists $l<j$ such that
    $S_j=S_l$ and $S_l=S'_l$\lixo{=f(S_l)}. Then
  \begin{gather*}
     \delta\djk{j}{k} \;=\; \delta\djk{l}{k} \;\;
\Rightarrow\;\;f(\delta\djk{j}{k})\;=\;f(\delta\djk{l}{k})\\
\Leftrightarrow\;\; \delta'\djk{j}{k}\;=\;\;\delta'\djk{l}{k}\lixo{
  \;\; \Leftrightarrow\; S'_j=S'_l\;=\;S_l=S_j.}
  \end{gather*}
Thus $S_j=S_l$ implies $S'_j=S'_l$, and so $S'_j=S_j$.
\item If $S_j=\max\{S_i\mid i<j\}+1$ then $S'_j\not\in\{S_i\mid i<j\}$
  because if there exists a $l<j$ such that $S'_l=S'_j$ by the same
  reason as before $S_j\in\{S_i\mid i<j\}$.  Thus, by \textbf{R1}
  $S'_j=\max\{S_i\mid i<j\}+1 = S_j$.
  \end{enumerate}

  Conversely, by Lemma~\ref{lem:ICDFA2}, we have that each string
  represents a \ICDFAE{} up to a compatible renaming of states, i.e.,
  if two \ICDFAE{}'s are represented by the same string, that
  representation defines a isomorphism between them.
\end{proof}
These string representations lead to a normal representation for
\ICDFAE{}'s. For each of them, if we add a sequence of \emph{final
  states}, we obtain a \textbf{normal form} for \ICDFA{}'s.

\section{Generating automata}
\label{sec:generating-automata}
Normal representations for \ICDFAE{}'s (as presented above) can be
used as compact computer representations for that kind of objects, but
even though rules \textbf{R1} and \textbf{R2} are quit simple, it is
not
evident how to write an enumerative algorithm in an efficient way.
In a string representing an \ICDFAE{} with $n$ states over an alphabet
of $k$ symbols, $[(S_i)_{i<kn}]$, let $(f_j)_{0<j<n}$ be the sequence
of indexes of the first occurrence of each state label~$j$. That those
indexes exist is a direct consequence of the way the string is
constructed. Now consider 
$b_1=f_1$, $b_j=f_j-f_{j-1}$, for $2\leq
j\leq n-1$ and $b_n=kn-f_{n-1}+1$. Note that
$\sum_{l=1}^{j}b_l=f_j$, for $2\leq j\leq n-1$.

Note that
$$  \sum_{l=1}^{j}b_l=f_j\text{ , for } j\in[n-1].$$
It is easy to see that
\begin{enumerate}
\item Rule \textbf{R1} simply states that
  \begin{equation*}
    (\forall 2 \leq j \leq n-1)(b_j>0).\tag{\textbf{G1}}\label{eq:6}
  \end{equation*}
\item Rule
  \textbf{R2} establishes that
  \begin{equation*}
    (\forall m\in[n-1])(f_m< km).\tag{\textbf{G2}}\label{eq:5}
\end{equation*}
\end{enumerate}
To generate all the automata, for each allowed sequence of
$(b_j)_{0<j<n}$ we can generate all the remaining symbols $S_i$ (those
with $i\not\in \{f_j\mid 0<j<n\}$) according to the following rules:
\begin{gather}
  i<b_1\;\Rightarrow\;S_i=0;\tag{\textbf{G3}}\\
  (\forall j\in[n-2]) (f_j<i<f_{j+1} \;\Rightarrow \;S_i\in[j]_0);
  \tag{\textbf{G4}}\\
  i>f_{n-1}\;\Rightarrow\; S_i\in[n-1]_0.\tag{\textbf{G5}}
\end{gather}

\section{Enumeration of \ICDFA{}'s}
\label{sec:enumeration}
In this section we obtain a formula $B_k(n)$ for the number of strings
$[(S_i)_{i<kn}]$ representing \ICDFAE{}'s with $n$ states over an
alphabet of $k$ symbols. Although it is already known a formula for
the number of non-isomorphic \ICDFAE{}'s, we think that our method is
new. Liskovets~\cite{liskovets69} and, independently,
Robinson~\cite{robinson85:_count} gave for that number the formula
$H_k(n)=\frac{h_k(n)}{(n-1)!}$ where $h_k(1)=1$ and for $n>1$
\begin{equation}
  \label{eq:hk}
h_k(n)=n^{kn}-\sum_{1\leq j<
  n}\binom{n-1}{j-1}n^{k(n-j)}h_k(j)
\end{equation}
Note that $n^{kn}$ is the number of transition functions, from which
we subtract the number of them that have $n-1$, $n-2$,\ldots,$1$
states not accessible from the initial state. And then, we may divide
by $(n-1)!$, as the names of the remaining states (except the initial)
are irrelevant. Reciprocally, the formula we will derive ($B_k(n)$) is
a direct positive summation.

First, let us consider the set of strings $[(S_i)_{i<kn}]$ with
$S_i\in[n-1]_0$ and satisfying only rule \textbf{R1}. The number of
these strings gives an upper bound for $B_k(n)$. This set can be given
by $A_{n}\cap[n-1]_0^{kn}$, where for $c>0$,
\begin{equation}
  \label{eq:an}
A_c=L(0^\star
+\sum_{i=1}^{c-1}0^\star\prod_{j=1}^{i}j(0+\cdots+j)^\star).
\end{equation}
These languages belong to a family of languages $L_c$ presented by
Moreira and Reis~\cite{moreira05} and that represent partitions of
$[n]$ with no more than $c\geq 1$ parts, i.e.,
\begin{equation}
  L_c=L(\sum_{i=1}^{c}\prod_{j=1}^{i}j(1+\cdots+j)^\star).\label{eq:7}
\end{equation}
We have that $\rho_{A_c}(n)=\rho_{L_c}(n+1)$ and that
$\rho_{L_c}(n)=\sum_{i=1}^{c}{S(n,i)}$, where $S(n,i)$ are Stirling
numbers of second kind. So we get that the number of strings of length
$kn$ that are in $A_n$, is $\rho_{A_n}(kn)=\sum_{i=1}^{n}{S(kn+1,i)}$.
We have the proposition,

\begin{prop}
 For all $n,\; k\geq 1$, $B_k(n)\leq\sum_{i=1}^{n}{S(kn+1,i)}$.
\end{prop}
For $n=3$ and $k=2$, $B_2(3)\leq 365$. For $k=2$, Bassino and
Nicaud~\cite{bassino:_enumer} presented a better upper bound, namely
that $B_2(n)\leq nS(2n,n)$.

Now let us consider only the rule \textbf{R2}. This rule can be formulated
as 
\begin{gather}
\bigwedge_{m=1}^{n-1}\bigvee_{j=0}^{km-1} S_j=m.
\end{gather}
From this formula it is easy to see that the strings $[(S_i)_{i<kn}]$
with $S_i\in[n-1]_0$ and satisfying only rule \textbf{R2} can be
represented by the regular expression
\begin{equation}
  \label{eq:strr2}
  \bigcap_{m=1}^{n-1} \sum_{j=0}^{km-2} (0+\cdots+(m-1))^jm(0+\cdots+
  (n-1))^{kn-j-1}, 
\end{equation}
\noindent where we extended the operators of regular expressions to
intersection.

Now in order to simultaneously satisfy rules \textbf{R1} and
\textbf{R2}, in formula~(\ref{eq:strr2}), the first occurence of $m$
must precede the one of $m-1$, for $2\leq m\leq n-1$. These positions
are exactly the sequence $(f_j)_{0<j<n}$ defined in
Section~\ref{sec:generating-automata}. Given these positions and
considering the correspondent sequence $(b_j)_{0<j<n}$ we obtain the
regular expression:
\begin{equation*}
  \label{eq:prodreg}
  \left (\prod_{j=1}^{n-1} (0+\ldots+(j-1))^{b_j-1}j
  \right)(0+\cdots+(n-1))^{b_n-1}, 
\end{equation*}

\noindent and we must consider the possible values of $(b_j)_{0<j<n}$,
constrained to~\textbf{G1} and \textbf{G2}:
\begin{equation*}
  \label{eq:regexp}
  \sum_{b_1=1}^{k} \sum_{b_2=1}^{2k-b_1}
    \cdots  \sum_{b_{n-1}=1}^{k(n-1)-\sum_{l=1}^{n-2}b_l} 
  \left(\prod_{j=1}^{n-1} (0+\ldots+(j-1))^{b_j-1}j
  \right)(0+\cdots+(n-1))^{b_n-1}
\end{equation*}


For $n=3$ and $k=2$ we have
\begin{equation*}
  \label{eq:reg23}
(01+1(0+1))((0+1)2+2(0+1+2))(0+1+2)^2+12(0+1+2)^4,
\end{equation*}
\noindent and the number of these strings is $(1+2)((2+3)3^2)+3^4= 216$.


For each sequence $(b_j)_{0<j<n}$ the number of strings
$[(S_i)_{i<kn}]$ with $S_i\in[n-1]_0$ and satisfying \textbf{R1} and
\textbf{R2} is 
\begin{equation}
  \label{eq:enu1}
  \prod_{j=1}^n j^{b_j-1},
\end{equation}
\noindent a direct consequence of rules \textbf{G3}, \textbf{G4} and
\textbf{G5}.  And then we must take the sums over all $b_j$
constrained to rules \textbf{G1} and \textbf{G2}.
\begin{thm}
  \label{thm:bkn}
 We have
  \begin{equation}
  \label{eq:bkn}
  B_k(n)=  \sum_{b_1=1}^{k} \sum_{b_2=1}^{2k-b_1}  \sum_{b_3=1}^{3k-b_1-b_2}
  \cdots  \sum_{b_{n-1}=1}^{k(n-1)-\sum_{l=1}^{n-2}b_l} 
  \prod_{j=1}^{n} j^{b_j-1}.
\end{equation}
  \end{thm}
\begin{proof}
  It is an immediate consequence of rules \textbf{G1} to \textbf{G5}.
\end{proof}

\begin{cor}\label{cor:icdfas}
  The number of non-isomorphic \ICDFA{}'s  with $n$ states over an
alphabet of $k$ symbols is 
$2^nB_k(n)$.
\end{cor}
\begin{proof}
  By Theorems~\ref{thm:iso} and~\ref{thm:bkn} and considering the
  possible sets of final states.
\end{proof}

\section{Conclusion}
\label{sec:implementation}
The method described in Section~\ref{sec:generating-automata} was
implemented
and used to generate all \ICDFAE{}'s for $k=2$ and $n<10$, and $k=3$
and $n<7$.  The time complexity of the program is linear in the number
of automata and took about a week to generate all the referred
\ICDFAE{}'s, in a PPC G4 1.5MHz.

One of the advantage of this method is that only the allowed strings
are computed so it is not a \emph{generate-and-test} algorithm and
because automata are generated in lexicographic order it is easy to
generate them as needed for consumption by another algorithm.

If an \ICDFA{} with $n$ states accepts a finite language then there
exists a topological order of its states such that
$\delta(i,\sigma)>i$, for all $i<n-1$ and $\sigma\in \Sigma$. But the
order we used for string representations is not a topological
order. So we can not determine directly from the string if the
accepted language is finite, as was done by
Domaratzki~\cite{domaratzki04:_combin_inter_gener_genoc_number} only
for finite languages. Although the formula $B_k(n)$ is quite similar
to the one obtained
in~\cite{domaratzki04:_combin_inter_gener_genoc_number} for an upper
bound of the number of finite languages, the meaning of the parameters
$(b_j)$ are not directely related.

\section{Acknowledgements}
We thank the anonymous referees for their
comments that helped to improve this~paper.

\end{document}